\newif\iflong
\newif\ifshort

 \longtrue

\iflong
\else
\shorttrue
\fi
\documentclass[letterpaper]{article}
\usepackage{aaai23} 
\usepackage{times} 
\usepackage{helvet} 
\usepackage{courier} 
\usepackage[hyphens]{url} 
\usepackage{graphicx} 
\usepackage{graphicx} 
\usepackage{natbib} 
\usepackage{caption} 
\nocopyright
\usepackage{algorithm}
\usepackage{algorithmic}

\usepackage{newfloat}
\usepackage{listings}
\DeclareCaptionStyle{ruled}{labelfont=normalfont,labelsep=colon,strut=off} 
\floatstyle{ruled}
\newfloat{listing}{tb}{lst}{}
\floatname{listing}{Listing}
\usepackage{tabularx}
\usepackage{ifthen}
\graphicspath{{./images/}}

\usepackage[utf8]{inputenc}
\usepackage{mathtools,amsthm,amssymb}
\usepackage{microtype,todonotes}
\usepackage{tikz}

\newcommand{\supp}{\operatorname{supp}}
\newcommand{\FPT}{\mathsf{FPT}}
\newcommand{\W}[1][]{
	\ifthenelse { \equal {#1} {} } 
    	{ \def\temp {\mathsf{W} } }   
    	{ \def\temp {\mathsf{W[#1]}} }   
    	\temp
}

\newcommand{\NP}{\ensuremath \mathsf{NP}}
\renewcommand{\P}{\ensuremath \mathsf{P}}
\newcommand{\XP}{\ensuremath \mathsf{XP}}
\newcommand{\bigoh}{\mathcal{O}}

\newcommand{\fptpactime}{\mathsf{FPT}\text{-}\mathsf{PAC_{time}}}
\newcommand{\fptpac}{\mathsf{FPT}\text{-}\mathsf{PAC}}
\newcommand{\xppactime}{\mathsf{XP}\text{-}\mathsf{PAC_{time}}}
\newcommand{\xppac}{\mathsf{XP}\text{-}\mathsf{PAC}}
\newcommand{\consis}[2]{\ensuremath \textsc{{$#1$}-Consistency({$#2$})}}
\newcommand{\fpt}{\mathrm{fpt}}
\newcommand{\poly}{\mathrm{poly}}
\newcommand{\xp}{\mathrm{xp}}

\newtheorem{thm}{Theorem}[section]
\newtheorem{obs}[thm]{Observation}

\newtheorem{lem}[thm]{Lemma}
\newtheorem{prop}[thm]{Proposition}
\newtheorem{claim}[thm]{Claim}
\newenvironment{claimproof}{\iflong \begin{proof}[{\upshape{\noindent\underline{Proof}}}]}{\end{proof}\fi}
\theoremstyle{definition}
\newtheorem{defn}[thm]{Definition}
\theoremstyle{remark}
\newtheorem{rem}[thm]{Remark}
\newtheorem*{ex}{Example}
\numberwithin{equation}{section}

\newcommand{\N}{\ensuremath{\mathbb{N}}}
\newcommand{\Nat}{\N}

\begin{document}
\title{A Parameterized Theory of PAC Learning}
\author{
Cornelius Brand\textsuperscript{\rm 1},
Robert Ganian\textsuperscript{\rm 1},
Kirill Simonov\textsuperscript{\rm 2}}

\affiliations {
\textsuperscript{\rm 1} Algorithms and Complexity Group, TU Wien, Austria \\
\textsuperscript{\rm 2}Chair for Algorithm Engineering, Hasso Plattner Institute, Germany \\
\{cbrand, rganian\}@ac.tuwien.ac.at, kirill.simonov@hpi.de
}

\maketitle
\begin{abstract}
Probably Approximately Correct (i.e., PAC) learning is a core concept of sample complexity theory, and efficient PAC learnability is often seen as a natural counterpart to the class $\P$ in classical computational complexity. But while the nascent theory of parameterized complexity has allowed us to push beyond the $\P$-$\NP$ ``dichotomy'' in classical computational complexity and identify the exact boundaries of tractability for numerous problems, there is no analogue in the domain of sample complexity that could push beyond efficient PAC learnability.

As our core contribution, we fill this gap by developing a theory of parameterized PAC learning  which allows us to shed new light on several recent PAC learning results that incorporated elements of parameterized complexity. Within the theory, we identify not one but two notions of fixed-parameter learnability that both form distinct counterparts to the class $\FPT$---the core concept at the center of the parameterized complexity paradigm---and develop the machinery required to exclude fixed-parameter learnability. We then showcase the applications of this theory to identify refined boundaries of tractability for CNF and DNF learning as well as for a range of learning problems on graphs.
\end{abstract}

\section{Introduction}
While a number of different models for sample complexity have by now been considered in the literature, the fundamental concept of Probably Approximately Correct (i.e., PAC) learning~\cite{Valiant84} remains a core pillar which can theoretically capture and explain the success of learning algorithms in a broad range of contexts. Intuitively, in the classical proper PAC setting we ask whether it is possible to identify (or ``learn'') a hypothesis from a specified hypothesis space after drawing a certain number of labeled samples\footnote{Formal definitions are provided in the Preliminaries.} according to an unknown distribution.
As the best case scenario, one would wish to have an algorithm that can learn a hypothesis with arbitrarily high precision and with arbitrarily high probability after spending at most polynomial time and using at most polynomially-many samples. These are called \emph{efficient} PAC-learning algorithms, and the class of learning problems admitting such algorithms forms a natural counterpart to the class $\P$ in the classical complexity theory of computational problems. 

In spite of this parallel, our understanding of the sample complexity of learning problems remains significantly behind the vast amounts of knowledge we have by now gathered about the time complexity of computational problems. Indeed, while understanding the distinction between $\NP$-hard and polynomially tractable classes of instances for computational problems remains an important research direction, over the past two decades focus has gradually shifted towards a more fine-grained analysis of the boundaries of tractability for such problems. In particular, \emph{parameterized complexity}~\cite{CyganFKLMPPS15,DowneyF13} has emerged as a prominent paradigm that yields a much deeper understanding of the limits of tractability than classical complexity theory, and today we see research exploring the parameterized complexity of relevant problems appear in a broad range of venues spanning almost all areas of computer science. The main idea behind the parameterized paradigm is to investigate the complexity of problems not only with respect to the input size $n$, but also based on a numerical parameter $k$ that captures some property of the input; the central notion is then that of \emph{fixed-parameter tractability} (FPT), which means that the problem of interest can be solved in time $f(k)\cdot n^{\bigoh(1)}$ for some computable function $f$.

Given its ubiquitous applications, it is natural to ask whether (and how) the principles of parameterized complexity analysis can also be used to expand our understanding of the foundations of sample complexity beyond what may be possible through the lens of efficient PAC-learnability.
In fact, there already are a handful of papers~\cite{LiL18,ArvindKL09,AlekhnovichBFKP08,GartnerG07} 
that have hinted at the potential of this novel field by identifying learning algorithms that seem to act as intuitive counterparts to the fixed-parameter algorithms emblematic of parameterized complexity. However, the theoretical foundations of a general parameterized extension to the PAC-learning framework have yet to be formalized and developed, stifling further development of this unique bridge between the two research communities. 

\smallskip
\noindent \textbf{Contribution.}\quad
As our first conceptual contribution, we lay down the foundations of a general parameterized extension to the PAC-learning framework. Crucially, when defining the analogue of fixed-parameter tractability in the sample complexity setting, we show that there are two natural definitions which need to be considered: a parameterized learning problem is
\begin{enumerate}
\item \emph{$\fpt$-time learnable} if a suitable PAC hypothesis can be learned from polynomially-many samples using a fixed-parameter algorithm, and
\item \emph{$\fpt$-sample learnable} if a suitable PAC hypothesis can be learned using a fixed-parameter algorithm.
\end{enumerate}

Essentially, in the latter case the total number of samples used need not be polynomial, but of course remains upper-bounded by the (fixed-parameter) running time of the algorithm. An important feature of the framework is that parameters are allowed to depend not only on the sought-after concept (as was considered in prior work~\cite{ArvindKL09}), 
but also on properties of the distribution.
This can be seen as analogy to how one utilizes parameters in the classical parameterized complexity paradigm: they can capture the properties we require on the output (such as the size of a solution), but also restrictions we place on the input (such as structural parameters of an input graph).

The parameterized PAC-learning framework proposed in this paper also includes a machinery that can be used to obtain lower bounds which exclude the existence of FPT-time and $\fpt$-sample learning algorithms; this is done by building a bridge to the well-studied $\W$-hierarchy in parameterized complexity theory.
Moreover, we provide the sample-complexity analogues to the complexity class $\XP$; similarly as in the fixed-parameter case, it is necessary to distinguish whether we restrict only the running time or the number of samples.

After laying down the foundations, we turn our attention to two major problem settings with the aim of illustrating how these notions can expand our knowledge of sample complexity. First, we study the boundaries of learnability for the fundamental problem of learning DNF and CNF formulas, which are among the most classical questions in PAC learning theory. 

On one hand, $k$-CNF and $k$-DNF formulas are known to be efficiently PAC-learnable for every fixed $k$, albeit the running time of the known algorithms have the form $\bigoh(n^{k})$. As our first result, we show that while both of these learning problems are $\xp$-time learnable, under well-established complexity assumptions they are neither $\fpt$-time nor $\fpt$-sample learnable when parameterized by $k$.
On the other hand, for each $p>1$
 it is well-known that $p$-term DNF and $p$-clause CNF formulas are not efficiently PAC learnable~\cite{PittV88,AlekhnovichBFKP08} unless $\P=\NP$, and here we ask the question whether a suitable parameterization can be used to overcome this barrier to tractability. We show that while learning $p$-term DNF and $p$-clause CNF formulas remains intractable under the most natural parameterizations of the target concept, one can learn such formulas by an $\fpt$-time algorithm by exploiting a parameterization of the distribution. In particular, it is not difficult to observe that the learning problem is tractable when each sample contains at most one variable set to \texttt{True}, and we show that it is possible to lift this to a non-trivial $\fpt$-time learning algorithm when parameterized by the number of variables which are exempt from this restriction (i.e., can always contain any value).

While CNF and DNF learning is fundamental for PAC learning, in the last part of our study we turn towards a setting which has been extensively studied in both the sample complexity and parameterized complexity paradigms: graph problems. Indeed, fundamental graph problems such as vertex cover and subgraph detection served as a testbed for the development of parameterized complexity theory, and are also heavily studied in machine learning and sample complexity contexts~\cite{NguyenM20,AbasiN19,AbasiBM18,DamaschkeM10,AlonA05,ChoiK10}. Here, we obtain a meta-theorem showing that a large class of graph problems is $\fpt$-time learnable: in particular, we provide an $\fpt$-time learning algorithm for finding a vertex deletion set to $\mathcal{H}$, where $\mathcal{H}$ is an arbitrary graph class that can be characterized by a finite class of forbidden induced subgraphs. As one special case, this captures the problem of learning hidden vertex covers~\cite{DamaschkeM10}. We conclude by excluding a similar result for classes $\mathcal{H}$ characterized by forbidden minors---hence, for instance hidden feedback vertex sets are neither $\fpt$-time nor $\fpt$-sample learnable when parameterized by their size.

\ifshort
\smallskip
\noindent \emph{Due to space constraints, proof details are provided in the supplementary material.}
\fi

\section{Preliminaries}
As basic notation and terminology, we set $\{0,1\}^\ast = \bigcup_{m\in\N} \{0,1\}^m$.
A \emph{distribution} on $\{0,1\}^n$ is a mapping $\mathcal{D}: \{0,1\}^n \rightarrow [0,1]$ such that $\sum_{x \in \{0,1\}^n} \mathcal{D}(x) = 1$, and the \emph{support} of $\mathcal{D}$ is the set $\supp\mathcal{D} = \{ x \mid \mathcal{D}(x) > 0\}$.
We first recall the basics of both PAC-learning and parameterized complexity theory.
\subsection{PAC-Learning}
Let us begin by formalizing the classical learning theory considered in this article~\cite{Valiant84,MLbook}.

\begin{defn}
A \emph{concept} is an arbitrary Boolean function $c: \{0,1\}^n \rightarrow \{0,1\}$.
An assignment $x\in \{0,1\}^n$ is called a \emph{positive sample} for $c$ if $c(x) = 1$, and a \emph{negative sample} otherwise. A \emph{concept class} $\mathcal C$ is a set of concepts.
For every $m\in \Nat$, we write $\mathcal C_m = \mathcal C \cap \mathcal{B}_m$, where $\mathcal{B}_m$ is the set of all $m$-ary Boolean functions.
\end{defn}

\begin{defn}
Let $\mathcal{C}$ be a concept class. A surjective mapping $\rho: \{0,1\}^\ast \rightarrow \mathcal{C}$ is called a \emph{representation scheme} of $\mathcal{C}$. 
We call each $r$ with $\rho(r) = c$ a \emph{representation} of concept $c$.
\end{defn}
In plain terms, while concepts may be arbitrary functions, representations are what make these ``tangible'' and are what we typically expect as the output of learning algorithms.

\begin{defn}
A \emph{learning problem} is a pair $(\mathcal{C}, \rho)$, where $\mathcal{C}$ is a concept class and $\rho$ is a representation scheme for $\mathcal{C}$.
\end{defn}

\begin{defn} \label{def:learning-alg}
A \emph{learning algorithm} for a learning problem $(\mathcal{C},\rho)$ is a randomized algorithm such that:
\begin{enumerate}
\item It obtains the values $n,\varepsilon,\delta$ as inputs, where $n$ is an integer and $0< \varepsilon,\delta \leq 1$ are rational numbers.
\item It has access to a hidden representation $r^\ast$ of some concept $c^\ast = \rho(r^\ast)$ and a hidden distribution $\mathcal{D}_n$ on $\{0,1\}^n$ through an oracle that returns \emph{labeled samples} $(x,c^\ast(x))$, where $x \in \{0,1\}^n$ is drawn at random from $\mathcal{D}_n$.
\item The output of the algorithm is a representation of some concept, called its \emph{hypothesis}.
\end{enumerate}
\end{defn}

\begin{rem}
Clearly, the algorithm can infer the value of $n$ from the samples that the oracle returns. Still, $n$ is included in the input for the sake of being explicit. We use $s=|r^\ast|$ to denote the size of the hidden representation.
\end{rem}

\begin{defn}
Let $\mathcal{A}$ be a learning algorithm. Fix a hidden hypothesis $c^\ast$ and a distribution on $\{0,1\}^n$.
Let $h$ be a hypothesis output by $\mathcal{A}$ and $c = \rho(h)$ be the concept $h$ represents.
We define 
\[
\operatorname{err}_h = \mathbb{P}_{x \sim \mathcal{D}_n}(c(x) \neq c^\ast(x))
\]
as the probability of the hypothesis and the hidden concept disagreeing on a sample drawn from $\mathcal{D}_n$,  the so-called \emph{generalization error} of $h$ under $\mathcal{D}_n$.

The algorithm $\mathcal{A}$ is called \emph{probably approximately correct (PAC)} if it outputs a hypo\-thesis $h$ such that $\operatorname{err}_h \leq \varepsilon$ with probability at least $1-\delta.$ 
\end{defn}

Usually, learning problems in this framework are regarded as tractable if they are PAC-learnable within polynomial time bounds. More precisely, we say that a learning problem $L$ is \emph{efficiently} PAC-learnable if there is a PAC algorithm for $L$ that runs in time polynomial in $n,s,1/\varepsilon$ and $1/\delta$.

\subsection{Parameterized Complexity}
To extend the above concepts from learning theory to the parameterized setting,
we now introduce the parts of parameterized complexity theory that will become important later.
We follow the excellent exposition of the textbook by Cygan et al.~(\citeyear{CyganFKLMPPS15}) and define the following central notions:
\begin{defn}
A \emph{parameterized problem} is a language $L \subseteq \{0,1\}^\ast \times \mathbb{N}$. For a pair $(x,k) \in \{0,1\}^\ast \times \mathbb{N}$, $k$ is called the \emph{parameter} of the instance $(x,k)$. The encoding length of $(x,k)$ is called the \emph{size} of the instance.
\end{defn}

\begin{rem}
The definition of parameterized decision problems can easily be extended to search problems, where we are additionally required to output a witness in case our algorithm outputs ``yes.''
\end{rem}

\begin{defn}
A parameterized problem $L$ is \emph{fixed-parameter tractable} if there exists an algorithm $\mathcal{A}$, a computable non-decreasing function $f: \mathbb{N} \rightarrow \mathbb{N}$ and a polynomially bounded non-decreasing function $p(\cdot,\cdot): \mathbb{N}^2 \rightarrow \mathbb{N}$ such that $\mathcal{A}$ correctly decides for an input $(x,k)$ whether or not $(x,k) \in L$ holds in time $f(k)\cdot p(n,k)$.
We denote the class of fixed-parameter tractable problems by $\FPT$.
\end{defn}

\begin{rem}
While fixed-parameter tractability lies at the heart of parameterized complexity theory, the class $\XP$ defined below captures a weaker notion of tractability that is still desirable for parameterized problems which are not believed to be in $\FPT$.
\end{rem}

\begin{defn}
A parameterized problem $L$ is in the complexity class $\XP$ if there exists an algorithm $\mathcal{A}$, a computable non-decreasing function $f: \mathbb{N} \rightarrow \mathbb{N}$ and a polynomially bounded non-decreasing function $p(\cdot,\cdot): \mathbb{N}^2 \rightarrow \mathbb{N}$ such that $\mathcal{A}$ correctly decides for an input $(x,k)$ whether or not $(x,k) \in L$ holds in time $p(n,k)^{f(k)}$.
\end{defn}
\begin{rem}
Less formally, $\XP$ is the class of problems that are solvable in polynomial time for a constant parameter value.
The difference to $\FPT$ is that for $\XP$, we allow this polynomial (and in particular, its degree) to depend on $k$,
while it is fixed for all $k$ in the definition of $\FPT$.
\end{rem}

\ifshort
We also assume basic familiarity with the complexity classes $\W[1]$ and $\W[2]$~\cite{DowneyF13}. It is a well-established conjecture that these are strictly larger than $\FPT$, and hence establishing $\W[1]$- or $\W[2]$-hardness for a problem essentially rules out its fixed-parameter tractability.
\fi

\iflong
To prove our lower bounds, we will to recall the basics of parameterized complexity theory,
beginning with the usual notion for parameterized reductions.
\begin{defn}
Let $L,L'$ be two parameterized problems.
We say that an algorithm $\mathcal{A}$ that takes an instance $(x,k)$ of $L$ and outputs an instance $(x',k')$ of $L'$ is a \emph{parameterized reduction from $L$ to $L'$} if
\begin{enumerate}
\item $(x,k) \in L$ if and only if $(x',k') \in L'$,
\item $k' \leq g(k)$ for some function $g: \N \rightarrow \N$, and
\item the running time of the algorithm is upper-bounded by $f(k)\cdot n^{\bigoh(1)}$ for some computable function $f: \N \rightarrow \N$.
\end{enumerate}
\end{defn}
Using this definition, we can give an ad-hoc definition of the parameterized complexity classes $\W[1]$ and $\W[2]$,
which form rough analogues to the usual class $\NP$. While we will only use $\W[2]$ in this paper,
we will give the definition for $\W[1]$ for completeness.

Let \textsc{$k$-Clique} be the problem of deciding, given a graph $G$ and a parameter $k$, if $G$ has a clique of size at least $k$.
Furthermore, we let \textsc{$k$-Hitting Set} be the problem of deciding, given a pair $(U,\mathcal{F})$ and a parameter $k$, where $U$ is some universe of the form $U = \{1,\ldots,n\}$ and $\mathcal{F}$ is a set family $\{F_i\}_{i=1}^m$ over $U$, if there is a set $H$ of size at most $k$ such that $F_i \cap H \neq \emptyset$ for all $i$.
\begin{defn}
The class $\W[1]$ contains all parameterized problems $L$ such that there is a parameterized reduction from $L$ to \textsc{$k$-Clique}. The class $\W[2]$ contains all parameterized problems $L$ such that there is a parameterized reduction from $L$ to \textsc{$k$-Hitting Set}.
Furthermore, we call a problem $L'$ \emph{$\W[1]$-hard} (resp. \emph{$\W[2]$-hard}) if there is a parameterized reduction from \textsc{$k$-Clique} (resp. \textsc{$k$-Hitting Set}) to $L'$.
\end{defn}

Establishing $\W[1]$- or $\W[2]$-hardness for a problem rules out its fixed-parameter tractability under well-established complexity assumptions~\cite{DowneyF13}.
\fi

\section{Parameterized PAC-Learning}
\label{sec:fptpac}

As the first step towards defining a theory of parameterized PAC learning, we need to consider the parameters that will be used. In the computational setting we associated a parameter with each instance, but in the learning setting this notion does not exist---instead, a learning algorithm needs to deal with a hidden concept representation and a hidden distribution, and we allow both of these to be tied to parameterizations.

\begin{defn}[Parameterization of Representations]
Let $\{\mathcal{R}_k\}_{k\in \N}$ with $\mathcal{R}_k \subseteq \{0,1\}^\ast$ be a mapping assigning a set of representations to every natural number $k$ such that for every $r \in \{0,1\}^\ast$, there is some $k$ such that $r \in \mathcal{R}_k$, that is, $\bigcup_k \mathcal{R}_k = \{0,1\}^\ast$.
We call $\{\mathcal{R}_k\}_{k\in\N}$ a \emph{parameterization of representations}.
Given a parameterization of representations $\{\mathcal{R}_k\}_{k\in \N}$, we associate a value $\kappa_{\mathcal{R}}(r)$ to single representations $r \in \{0,1\}^\ast$ by defining
\[
\kappa_{\mathcal{R}}(r) = \min \{k\ :\ r \in \mathcal{R}_k\}.
\]
\end{defn}

\begin{rem}
In line with the usual notion of parameterizations, we will assume $\kappa_{\mathcal{R}}$ to be a computable function.
\end{rem}

\begin{ex}
Let $\rho$ be the representation scheme taking a (binary representation of a) $k$-term DNF formula to its underlying Boolean function.
If we are interested in learning $k$-term DNFs, we probably want to consider a $k'$-term DNF a $k$-term DNF for $k' \leq k$, too. Therefore, we let $\mathcal{R}_k$ be the set of all $k'$-term DNFs for $k' \leq k$. The associated mapping $\kappa_{\mathcal{R}}$ then maps every $k$-term DNF to the value of $k$.
\end{ex}

\begin{defn}[Parameterization of Distributions and Samples]
Let $\lambda$ be a mapping assigning a natural number to every distribution on $\{0,1\}^n$ for each $n$, such that for every two distributions $\mathcal{D},\mathcal{D'}$ on $\{0,1\}^n$, if $\supp \mathcal{D} \subseteq \supp \mathcal{D'}$, then $\lambda(\mathcal{D}) \leq \lambda(\mathcal{D'})$.
In this case, we call $\lambda$ a \emph{parameterization of distributions}.

We extend every parameterization of distributions $\lambda$ to subsets $X \subseteq \{0,1\}^n$ by defining its corresponding \emph{parameterization of sample sets} via: 
\begin{align} \label{eq:param-sample}
\lambda(X) = \min_{\mathcal{D}\ :\ X\ =\ \supp \mathcal{D}} \lambda(\mathcal{D}).
\end{align}
\end{defn}

\begin{rem}
Note that the definition implies that, equivalently, $\lambda(X) = \min_{\mathcal{D}\ :\ X\ \subseteq \ \supp \mathcal{D}} \lambda(\mathcal{D})$; in other words, $\lambda(X)$ is the lowest parameter that can be obtained from a distribution that has $X$ as its support.
\end{rem}

All parameterizations considered in our exposition depend solely on the support; however, we do want to explicitly also allow more expressive parameterizations that depend on the distribution. To build a theory allowing such parameterizations, it is necessary to impose an additional technical condition which ensures that the distributions minimizing the parameter values for $\lambda$ are ``well-behaved''. 

\begin{defn}
We say that a distribution $\mathcal{D}^\ast$ is called \emph{typical for $X$ under $\lambda$} if 
$\mathcal{D}^\ast$ attains the minimum in Eq. \eqref{eq:param-sample}, that is, $\lambda(X) = \lambda(\mathcal{D}^\ast)$ and $\supp \mathcal{D}^\ast = X$.
\end{defn}

\begin{defn}[$L$-Sampleable Parameterizations]\label{def:sampleabe}
Let $\mathcal{A}$ be a randomized algorithm that receives input $X \subseteq \{0,1\}^n$ and outputs a concept $c$ in time $L(n,|X|,\lambda(X))$ for some non-decreasing function $L(\cdot,\cdot,\cdot)$. We say that $\lambda$ is \emph{$L$-sampleable} if the following holds true for the random variable $C$ that corresponds to the output of $\mathcal{A}$ over all random bits used by $\mathcal{A}$ to output $c$: There is some distribution $\mathcal{D}^\ast$ that is typical for $X$ under $\lambda$ such that $C \sim \mathcal{D}^\ast$ holds, i.e., $C$ has the same distribution as $\mathcal{D}^\ast$.
\end{defn}

\begin{rem}
Every parameterization $\lambda$ that depends only on the support $X$ is linear-time-sampleable, since the uniform distribution on $X$ will be typical for $X$ under $\lambda$. 
More generally, we believe that every ``reasonable'' parameterization is polynomial-time-sampleable.
\end{rem}
\begin{ex}
For instance, a valid choice of $\lambda$ is the mapping that selects the largest number $k$ such that there is a concept $x$ in the support of $\mathcal{D}$ that has $k$ non-zero entries.
This parameter can obviously be computed on sets of concepts, as demanded by the condition, and is linear-time sampleable  (in the size of the input, $n\cdot t$).
\end{ex}

\begin{defn}[Parameterized Learning Problems]
A \emph{parameterized learning problem} is a learning problem $(\mathcal{C},\rho)$ together with a pair $(\{\mathcal{R}_k\}_{k\in\N}, \lambda)$, called its \emph{parameters}, where $\{\mathcal{R}_k\}_{k\in\N}$ is a parameterization of representations and $\lambda$ is a parameterization of distributions. 

\end{defn}

At this point, a note on the relation between the definitions presented above and the formalism developed by \citeauthor{ArvindKL09} (\citeyear{ArvindKL09}) is in order. 
While the latter is capable of expressing learning problems such as $k$-juntas, where $k$ is the parameter, it cannot account for properties of the sample space. For example, one might be interested in learning Boolean formulas from samples that have a limited number of $k$ variables set to true. Our framework captures this by including the distribution into the parameterization. 
This can be likened on a conceptual level to the difference that exists in ordinary parameterized complexity theory between parameterizations by solution properties versus parameterizing by properties of the input instance.

\begin{defn}[Parameterized Learning Algorithm]
A \emph{parameterized learning algorithm} for a parameterized learning problem $(\mathcal{C},\rho,\{\mathcal{R}_k\}_{k\in\N},\lambda)$ is a learning algorithm for $(\mathcal{C},\rho)$ in the sense of Definition \ref{def:learning-alg}.
In addition to $n, \varepsilon, \delta$, a parameterized learning algorithm obtains two inputs $k$ and $\ell$, which are promised to satisfy $k = \kappa_{\mathcal{R}}(r^\ast)$ as well as $\ell= \lambda(\mathcal{D}_n)$,
and the algorithm is required to always output a hypothesis $h$ satisfying $\kappa(h) \in \mathcal{R}_k$.
\end{defn}
\begin{rem}
By requiring the hidden hypothesis and the output hypothesis to adhere to the same representation scheme, we limit ourselves to the setting of \emph{proper} learning. In principle, nothing speaks against extending our framework also to the improper case, as is done in \citeauthor{ArvindKL09} (\citeyear{ArvindKL09}) for their formalization.
Since all our examples speak about proper learning, we restrict our definitions to this case.
\end{rem}

\begin{rem}
As readers acquainted with parameterized complexity theory may find noteworthy, parameterized learning problems as defined here depend on two parameters, as opposed to a single parameter.
For parameterized algorithms, it is customary to combine multiple parameters $k,\ell$ into one via defining a new parameter such as $k + \ell$ or $\max\{k,\ell\}$. Indeed, this leads to less heavy notation, while sacrificing nothing in terms of expressive power of the obtained theory. 

We shall see towards the end of this section that for the purposes of this article, it is more convenient to make do with two separate parameters, in order to establish a meaningful link between ordinary parameterized algorithms and parameterized learning algorithms. 
\end{rem}

Let $\poly(\cdot)$ denote the set of functions that can be bounded by non-decreasing polynomial functions in their arguments.
Furthermore, $\fpt(x_1,\ldots,x_t;k_1,\ldots,k_t)$ and $\xp(x_1,\ldots,x_t;k_1,\ldots,k_t)$ denote those functions bounded by $f(k_1,\ldots,k_t)\cdot p(x_1,\ldots,x_t)$ and $p(x_1,\ldots,x_t)^{f(k_1,\ldots,k_t)}$, respectively, for any non-decreasing computable function $f$ in $k_1,\ldots,k_t$ and $p \in \poly(x_1,\ldots,x_t)$.

\begin{defn}[$(T,S)$-PAC Learnability] \label{def:learnable}
Let $T(n,s,1/\varepsilon,1/\delta,k,\ell),S(n,s,1/\varepsilon,1/\delta,k,\ell)$ be any two functions taking on integer values, and non-decreasing in all of their arguments.

A parameterized learning problem $\mathcal{L} = (\mathcal{C},\rho, \{\mathcal{R}_k\}_{k\in\N},\lambda)$ is \emph{$(T,S)$-PAC learnable} if there is a PAC learning algorithm for $\mathcal{L}$ that runs in time $\bigoh(T(n,s,1/\varepsilon,1/\delta,k,\ell))$ and queries the oracle at most $\bigoh(S(n,s,1/\varepsilon,1/\delta,k,\ell))$ times. 

We denote the set of parameterized learning problems that are $(T,S)$-PAC learnable by $\mathsf{PAC}[T,S]$. This is extended to sets of functions $\mathbf{S},\mathbf{T}$ through setting $\mathsf{PAC}[T,S] = \bigcup_{S \in \mathbf{S},\\T \in \mathbf{T}} \mathsf{PAC}[T,S]$.
\end{defn}

\begin{defn} \label{def:classes}
We define the following complexity classes:
\begin{align}
\fptpactime &= \mathsf{PAC}[\mathrm{fpt},\mathrm{poly}], \\
\fptpac &= \mathsf{PAC}[\mathrm{fpt},\mathrm{fpt}], \\
\xppactime &= \mathsf{PAC}[\mathrm{xp},\mathrm{poly}], \\
\xppac &= \mathsf{PAC}[\mathrm{xp},\mathrm{xp}],
\end{align}
where we fixed
\begin{align*}
\poly &= \poly(n,s,1/\varepsilon,1/\delta,k,\ell), \\
\fpt &= \fpt(n,s,1/\varepsilon,1/\delta;k,\ell), \\
\xp &= \xp(n,s,1/\varepsilon,1/\delta;k,\ell).
\end{align*}
\end{defn}
\begin{rem}
In addition to the complexity classes just defined, there is a fifth class that may be considered here: $\mathsf{PAC}[\mathrm{xp},\mathrm{fpt}].$  However, this class does not play a role in any of the examples presented in this work, and hence we leave its exploration to future works.
\end{rem}
Figure \ref{fig:classes} provides an overview of these complexity classes and their relationships. As does $\XP$, the class $\xppac$ in the learning setting contains precisely those parameterized learning problems which become efficiently PAC-learnable whenever the parameters are fixed to an arbitrary constant.
\begin{ex}
A problem that fits into the class $\xppac$ is learning of $k$-CNF and $k$-DNF formulas,
parameterized by $k$.
On the other hand, the upcoming Theorems \ref{thm:few_columns} and~\ref{thm:induced_subgraphs} furnish the class $\fptpactime$.
An example for the class $\fptpac$ in spirit can be found in \citeauthor{LiL18} (\citeyear{LiL18}). Finally, an example of a learning problem in $\xppactime$ is provided in Observation~\ref{obs:xpgraphs}.
\end{ex}
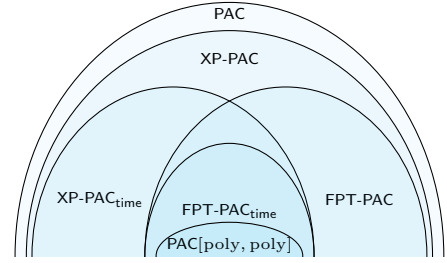
\begin{figure}
\begin{center}

\begin{tikzpicture}[set/.style={fill=cyan,fill opacity=0.04},scale=0.75]
\clip (-4.3,0) rectangle (5,5); 
\draw[set] (0,0) ellipse (3.8cm and 4.5cm);
\node at (0,4.3) {\tiny $\mathsf{PAC}$};

\draw[set] (0,0) ellipse (3.6cm and 4cm);
\node at (0,3.5) {\tiny $\xppac$};
 
\node at (-2.3,1) {\tiny $\xppactime$};
\draw[set] (-1,0) ellipse (2.5cm and 3cm);

\node at (2.3,1) {\tiny $\fptpac$};
\draw[set] (1,0) ellipse (2.5cm and 3cm);
 
\node at (0,0.8) {\tiny $\fptpactime$};
\draw[set] (0,0) ellipse (1.5cm and 2cm);

\node at (0,0.2) {\tiny $\mathsf{PAC}[\poly,\poly]$};
\draw[set] (0,0) ellipse (1.3cm and 0.6cm);
\end{tikzpicture}
\end{center}
\caption{A schematic view of the parameterized learning classes defined in Definition \ref{def:classes}.}
\label{fig:classes} 
\end{figure}

\subsection{Consistency Checking: A Link between Theories}
With the basic complexity classes in place, we turn to establishing links between the newly developed theory and the by now well-established parameterized complexity paradigm for decision problems. Crucially, these links will allow us to exploit algorithmic upper and lower bounds from the latter also in the learning setting.

\begin{defn}[Parameterized Consistency Checking] \label{def:decision-from-learning}
With every parameterized learning problem $\mathcal{L} = (\mathcal{C},\rho,\{\mathcal{R}_k\}_{k\in\N},\lambda)$ and every function $f(n,t,k,\ell)$ we associate a parameterized search problem $\consis{f}{\mathcal{L}}$
as follows:

\begin{enumerate}
\item Its input is a list of labeled samples $((x_1,a_1),\ldots,(x_t,a_t))$, with $x_i \in \{0,1\}^n$ for all $i$ and $x_i$ pairwise distinct, as well as $a_i \in \{0,1\}$. 

\item Its parameters are $k \in \N$, given as part of the input, and $\ell = \lambda(\{x_1,\ldots,x_t\})$. 

\item The task is to decide whether there is a representation $r$ that has $\kappa(r) \in \mathcal{R}_k$ with $|r| \leq f(n,t,k,\ell)$ and for the concept $c = \rho(r)$, it holds that $c(x_i) = a_i$ for all $i$.  
\end{enumerate}
We call this parameterized problem the \emph{parameterized consistency checking problem} associated with $\mathcal{L}$ and $f$.
\end{defn}

Ignoring all parameters in Definition \ref{def:decision-from-learning} gives the usual notion of consistency checking problems~\cite{PittV88}. 
The function $f$ is required to provide an explicit upper-bound on the sought-after representation; indeed, while each learning instance comes with a hidden representation of a certain size, this is not the case with the consistency checking problem and we need to allow the running time bounds to take the size of the representation into consideration.

It is well-known that, under the assumption that the hypothesis space is not too large, there is an equivalence between a learning problem being PAC-learnable and the corresponding decision problem being solvable in randomized polynomial time.
We now observe that a similar equivalence holds also in the parameterized setting.
\begin{lem} \label{lem:red-dec-to-learn}
Let $\mathcal{L} = (\mathcal{C},\rho,\{\mathcal{R}_k\}_{k\in \N} ,\lambda)$ be a parameterized learning problem,
and let $f(n,t,k,\ell)$ be a function.
If $\mathcal{L}$ is $(T,S)$-PAC learnable and $\lambda$ is $L$-sampleable, then there exists a randomized algorithm that, upon input labeled samples $\{(x_1,a_1),\ldots,(x_t,a_t)\}$ and parameters $k,\ell$ with $\lambda(\{x_1,\ldots,x_t\}) = \ell$, returns a consistent hypothesis $h \in \mathcal{R}_k$ of size $|h| \leq s = f(n,t,k,\ell)$ with probability at least $1-\delta$ if it exists, and runs in time $\bigoh(T(n, s, t, \delta, k, \ell) + L(n,t,\ell) \cdot S(n, s, t,\delta,k,\ell))$.
\end{lem}
\iflong \begin{proof}
Let $((x_1,a_1),\ldots,(x_t,a_t)),k,\ell$ be some inputs to the parameterized consistency problem associated with $\mathcal{L}$ such that $\lambda(\{x_1,\ldots,x_t\}) = \ell$.
We set $\varepsilon = \frac{1}{t+1}$, and take $\delta$ as desired for the success probability.

We run the $(T,S)$-PAC algorithm with inputs $n = |x_1| = \ldots = |x_t|$, $\varepsilon$ and $\delta$.
During execution of this algorithm, we emulate the oracle to a hidden distribution by calling the sampling algorithm guaranteed to exist by $\lambda$.

By definition, the PAC-learning algorithm will return a hypothesis from $\mathcal{R}_k$ inconsistent with at most a $1/(t+1)$-fraction of the samples, and therefore consistent with all of them, with probability at least $1-\delta$ (if it exists), which we output. Otherwise, we output some arbitrary hypothesis.
This gives a randomized algorithm for the parameterized consistency checking problem satisfying the required bound on the success probability.

As for the running time, note that we call the learning algorithm with the parameters as in the statement of the Lemma, which takes $T(n, s, t, \delta, k, \ell)$ steps, and for each of the $S(n, s, t,\delta,k,\ell)$ oracle calls the learning algorithm makes, we need to spend time $L(n,t,\ell)$ to produce a sample from the distribution.
We note that, if the PAC algorithm were to spend more than $T(n,s,t,k,\ell)$ steps (or draw more than $S(n,s,t,k,\ell)$ samples), we may safely halt its execution and return an arbitrary hypothesis\footnote{To be completely precise, the functions would have to meet the criterion of being asymptotically equivalent to some time-constructible function. Since this plays no role whatsoever in any natural scenario, we chose to ignore these details.}.
\end{proof}\fi

\begin{thm} \label{thm:red-dec-to-learn}
Let $\mathcal{L} = (\mathcal{C},\rho,\{\mathcal{R}_k\}_{k\in \N},\lambda)$ be a parameterized learning problem,
and let $f(n,t,k,\ell)$ be a function.

If $\mathcal{L}$ is in $\fptpac$, $f \in \fpt(n,t;k,\ell)$ and $\lambda$ is $S$-sampleable for some $S\in \fpt(n,t;\ell)$, then the parameterized consistency checking problem $\consis{f}{\mathcal{L}}$ associated with $\mathcal{L}$ and $f$ is in $\FPT$.

Similarly, if $\mathcal{L}$ is in $\xppac$, $f \in \xp(n,t;k,\ell)$ and $\lambda$ is $S$-sampleable for some $S \in \xp(n,t;\ell)$,
then the parameterized consistency checking problem $\consis{f}{\mathcal{L}}$ associated with $\mathcal{L}$ and $f$ is in $\XP$.
\end{thm}
\iflong \begin{proof}
Follows by inserting the appropriate functions into $T,S$ and $L$ in Lemma \ref{lem:red-dec-to-learn}.
Note, in particular, that the dependency on $s = f(n,t,k,\ell)$ in $T$ and $S$ is only allowed to be of the form $g(k,\ell)\cdot \poly(n,s,1/\varepsilon,1/\delta)$ or $\poly(n,s,1/\varepsilon,1/\delta)^{g(k,\ell)}$, by the respective definitions of $\fptpac$ and $\xppac$. Hence, if $f \in \fpt(n,t;k,\ell)$, then $\poly(n,f(n,t;k,\ell),1/\varepsilon,1/\delta)$ is still bounded by $g'(k,\ell)\cdot\poly(n,t,1/\varepsilon,1/\delta)$, and similarly for $f \in \xp(n,t;k,\ell)$.
\end{proof}\fi

\begin{lem} \label{lem:red-learn-to-dec}
Let $\mathcal{L} = (\mathcal{C},\rho,\{\mathcal{R}_k\}_{k\in \N},\lambda)$ be a parameterized learning problem, and let $\mathcal{H}_{n,k} = \mathcal{R}_k \cap \rho^{-1}(\mathcal{C}_n)$ be the set of all representations under $\rho$ in $\mathcal{R}_k$ of concepts in $\mathcal{C}_n$.
Suppose there is a deterministic algorithm running in time $T(n, t, \delta, k, \ell)$ for the consistency checking problem $\consis{\log|\mathcal{H}_{n,k}|}{\mathcal{L}}$. Then, $\mathcal{L}$ is $(T',S)$-PAC learnable,
where\iflong\footnote{The dependency on $s$ disappears because $s$ itself is bounded by a function of $n$ and $k$.}\fi
\begin{align*}
S(n,s,1/\varepsilon,1/\delta,k,\ell) &= \frac{1}{\varepsilon}(\log |\mathcal{H}_{n,k}| + \frac{1}{\delta})),\\
T'(n,s,1/\varepsilon,1/\delta,k,\ell) &= T(n,s,\frac{1}{\varepsilon}(\log |\mathcal{H}_{n,k}| + \frac{1}{\delta}),k,\ell).
\end{align*}
\end{lem}
\iflong \begin{proof}
The learning algorithm proceeds by drawing $t = \frac{1}{\varepsilon}(\log |\mathcal{H}_{n,k}| + \frac{1}{\delta})$ samples $(x_1,a_1),\ldots,(x_t,a_t)$ from the unknown distribution $\mathcal{D}$, and calls the algorithm using these $t$ samples and the parameter $k$. It outputs the hypothesis $h$ returned by the algorithm, which is guaranteed to be in $\mathcal{R}_k$, and satisfies $|h| \leq \log |H_{n,k}|$.
By a standard argument via the union bound, the resulting algorithm is indeed a parameterized PAC algorithm.

For this call to the algorithm, note that by definition, $\lambda(\{x_1,\ldots,x_t\}) \leq \ell$.
Recall from Definition \ref{def:learnable} that we assumed $T$ and $S$ to be non-decreasing.
Hence, by assumption, this takes time at most $T(n,\frac{1}{\varepsilon}(\log |\mathcal{H}_{n,k}| + \frac{1}{\delta}),k,\ell),$ and uses $t = \frac{1}{\varepsilon}(\log |\mathcal{H}_{n,k}| + \frac{1}{\delta})$ samples.
\end{proof}\fi

\begin{thm}\label{thm:red-learn-to-dec}
Let $\mathcal{L} = (\mathcal{C},\rho,\{\mathcal{R}_k\}_{k\in\N},\lambda)$, and denote the set of representations of $\mathcal{C}_n$ in $\mathcal{R}_k$ under $\rho$ as $\mathcal{H}_{n,k}$.

If the parameterized consistency checking problem $\consis{\log |\mathcal{H}_{n,k}|}{\mathcal{L}}$ is in $\FPT$ and $\log |\mathcal{H}_{n,k}| \in \fpt(n;k)$,
then $\mathcal{L}$ is in $\fptpactime$.

Similarly, if the parameterized consistency checking problem  $\consis{\log |\mathcal{H}_{n,k}|}{\mathcal{L}}$ is in $\XP$ and $\log |\mathcal{H}_{n,k}| \in \xp(n;k)$, then $\mathcal{L}$ is in $\xppactime$.
\end{thm}
\iflong \begin{proof}
Follows from Lemma \ref{lem:red-learn-to-dec} by and then inserting the respective running time functions and the bounds on $\mathcal{H}_{n,k}$, which in particular bound the length of any potentially returned hypotesis.
\end{proof}\fi

The significance of the previous theorems lies in transferring parameterized algorithmic upper and lower bounds for consistency checking into upper and lower bounds for parameterized learning problems, respectively.

To wit, Theorem \ref{thm:red-dec-to-learn} allows us to conclude from the fact that a parameterized consistency checking  problem is efficiently solvable by a parameterized algorithm that also the parameterized learning problem it belongs to is efficiently solvable. This will be exploited in Theorems~\ref{thm:few_columns},~\ref{thm:induced_subgraphs}.

On the other hand, Theorem \ref{thm:red-learn-to-dec} tells us that an efficient algorithm for a parameterized learning problem implies an efficient algorithm for the corresponding paramerized consistency checking problem. Turning this around, we see that lower bounds on consistency checking imply lower bounds for learning. This will be exploited in Theorems \ref{thm:k_dnf_tight},~\ref{thm:graphnotPAC}.

\section{Previous Work on Parameterized Learning}
After having introduced this new framework, we dedicate a separate section to how it fits in 
with the few papers that have considered parameterized approaches to learning theory,
and how it differs from them. 

Proceeding chronologically, the earliest link between parameterized complexity theory and learning theory was established already by the pioneers of paramerized complexity \cite{DowneyEF93}. 
While they also establish a parameterized learning model and link learning complexity with well-studied graph problems, in contrast to our work, they study exact learning by extended equivalence queries.
\citeauthor{GartnerG07} (\citeyear{GartnerG07}) study the complexity of learning directed cuts in various learning models, including PAC learning.
While their respective result falls under the standard poly-time PAC learning, it is notable that the improved learning bound they show is in terms of a structural parameter of the graph.
Better-than-trivial exponential running-time bounds for learning $k$-term DNF formulas have been studied, e.g., $n^{\tilde{\mathcal{O}}(\sqrt{n \log k})}$~\cite{AlekhnovichBFKP08}.
We already compared our framework to the earlier work by  \citeauthor{ArvindKL09} (\citeyear{ArvindKL09}) in the previous section, which only identified the class $\fptpac$.

The more recent work of \citeauthor{LiL18} (\citeyear{LiL18}) on parameterized algorithms for learning mixtures of linear regressions is not, strictly speaking, within our framework of parameterized PAC learning, since their error parameter is the distance to the hidden concept rather than the probability of mislabelling under any distribution. However, their sample complexity bound is $\FPT$, turning into infinite when the parameters are not bounded, and it is also a notable example of a PAC-learning result that holds only for a certain family of distributions (an approach which is also covered by our framework; see the intuition behind Definition~\ref{def:sampleabe}).
Finally, \citeauthor{BergeremGR22} (\citeyear{BergeremGR22}) consider parameterized complexity of learning first-order logic; they show hardness results via parameterized reductions as well as $\FPT$-time learning via consistency checking, which aligns well with our framework.

\section{Parameterized DNF and CNF Learning}

For the most of this section, we consider the problem of learning a hidden DNF formula under various parameterizations, which are formally defined throughout the section. We first observe that the problem of learning a CNF formula is equivalent to learning a DNF formula under any parameterization that is preserved under negation of the formula: indeed, by negating the input to a DNF-learning algorithm and then negating the output formula we obtain a CNF-learning algorithm, and vice versa. Thus our results hold for learning CNF formulas as well, which we do not show explicitly below.

Let \textsc{Learning DNF} be the learning problem $(\mathcal{C}, \rho)$ where $\mathcal{C}$ is the set of boolean functions corresponding to formulas in disjunctive normal form and $\rho$ is the straightforward representation of the terms of the formula.
To define a fundamental variant of the problem where the number of terms in the target formula is bounded, let \textsc{Learning $k$-term DNF} be the parameterized learning problem $(\mathcal{C}, \rho, \{\mathcal{R}_k\}_{k\in\N}, \lambda)$ where $\mathcal{C}$ and $\rho$ are as above, $\kappa$ maps the representation of the formula to the number of terms in it, and $\lambda$ is any constant parameterization.
Alas, it is well-known that the consistency checking problem for learning even $2$-term DNF is $\NP$-hard~\cite{AlekhnovichBFKP08}, and thus by a non-parameterized variant of Theorem~\ref{thm:red-dec-to-learn}, \textsc{Learning $k$-term DNF} is not in $\xppactime$ unless $\P=\NP$. Thus, in this parameterization one should not expect any tractability results.

Another natural way to parameterize \textsc{Learning DNF} would be to bound the maximum length of a term. Specifically, let \textsc{Learning $k$-DNF} be the parameterized learning problem $(\mathcal{C}, \rho, \{\mathcal{R}_k\}_{k\in\N}, \lambda)$ where $\mathcal{C}$ and $\rho$ are as in \textsc{Learning DNF}, $\kappa = \kappa_{\mathcal{R}}$ maps the representation of the formula to the maximum length of a term in it, and $\lambda$ is trivial. It is well-known that, for every fixed $k$,  \textsc{Learning $k$-DNF} is efficiently PAC-learnable by a brute-force argument and hence the problem is immediately in $\xppactime$ parameterized by $k$~\cite[Example 2.9]{MLbook}.
On the other hand, we show that under standard parameterized complexity assumptions this learning result is tight, extending the hardness reduction of Arvind et al.~\cite{ArvindKL09} for $k$-juntas and --monomials.

\begin{thm}
     Assuming $\W[2] \ne \FPT$, \textsc{Learning $k$-DNF} and \textsc{Learning $k$-CNF} are not in $\fptpactime$.
    \label{thm:k_dnf_tight}
\end{thm}
\iflong \begin{proof}
    We show the result for \textsc{Learning $k$-CNF}. Assume the contrary, then by Theorem~\ref{thm:red-dec-to-learn} there is an $\FPT$ algorithm for the consistency checking problem. We now show that there is a parameterized reduction from \textsc{$k$-Hitting Set}  to the consistency checking problem for \textsc{Learning $k$-CNF} even for $f$ linear in $n$ and $k$, reaching a contradiction.

    Consider an instance $(U, \mathcal{F})$ of \textsc{$k$-Hitting Set}. Set $n = |U|$, $t = |\mathcal{F}| + 1$, for each $i \in [t - 1]$, $x_i[j] = 1$ iff $j$-th element of the universe belongs to the $i$-th set of $\mathcal{F}$, and $a_i = 1$. Finally, set $x_t$ to be a zero vector and $a_t = 0$.

    First, consider a $k$-hitting set $S \in U$. Consider a one-clause $k$-CNF formula $\varphi$ that contains positive literals corresponding to the set $S$. For each $i \in [t - 1]$, $\varphi(x_i) = 1 = a_i$ since $S$ intersects the $i$-th set of $\mathcal{F}$, and the corresponding element of $x_i$ is set to $1$. Moreover, $\varphi(x_{t}) = 0 = a_{t}$ since $\varphi$ contains no negative literals.

    In the other direction, consider a $k$-CNF formula $\varphi$ that labels all the samples correctly. Since $\varphi(x_{t}) = a_{t} = 0$, there exists a clause $C$ in $\varphi$ that only contains positive literals. For each $i \in [t - 1]$, $\varphi(x_i) = a_i = 1$, and so $C(x_i) = 1$, thus there exists a positive literal in $C$ which variable corresponds to an element of the $i$-th set of $\mathcal{F}$. Thus the set $S$ consisting of elements corresponding to variables in $C$ is a $k$-hitting set of $\mathcal{F}$.
\end{proof}\fi

The above covers the two most natural parameterizations of the hypothesis class --- the number of terms and the maximum size of a term in the target formula --- however, we did not yet touch parameterizations of the sample space. As a simple step, observe that solving the consistency problem is easy when each assignment assigns at most one variable to ``true''; a one-term DNF can always be produced in this case. On the other hand, the previously-known coloring reduction~\cite{PittV88} shows that even when each assignment has at most two variables assigned to ``true'', the problem becomes $\NP$-hard. In the following, we analyze the tractability of learning $k$-term DNF \ifshort and $k$-clause CNF \fi ``between'' these two cases, and show the following positive result based on parameterizing by a notion of ``backdoor to triviality''~\cite{GaspersS12,SemenovZOKI18}.

\iflong
\begin{lem}
    \label{lem:few_columns}
    \textsc{Learning $k$-Term DNF} is $\fptpactime$ parameterized by $k + s$, where $s$ is the minimum size of a set $S$ of variables such that the support of the distribution satisfies the following:
    \begin{itemize}
    \item each assignment assigns at most one variable outside of $S$ to ``true'', and 
    \item each of the remaining variables, at most one assignment assigns it to ``true''.
    \end{itemize}
\end{lem}
\begin{proof}
    By Theorem~\ref{thm:red-learn-to-dec}, it is enough to provide an $\FPT$ algorithm for the consistency checking problem with $f = \fpt(n;k+s)$; in particular, it suffices to argue that we can bound the length of all possible hypotheses by such an $f$. 
    
    Partition the given set of labelled samples $(x_1, a_1)$, \ldots, $(x_t, a_t)$ into $b$ classes where all assignments in the same class agree on $S$: $\{(x_i, a_i)\}_{i \in [t]} = \{(x_i^1, a_i^1)\}_{i \in [t_1]} \cup \cdots \cup \{(x_i^b, a_i^b)\}_{i \in [t_b]}$, where for each $i, i' \in [t_j]$, $\left.x_i^j\right|_S = \left.x_{i'}^j\right|_S$. We now show that the instance can always be reduced to a smaller one, unless each class is of size at most $k$. The next two claims provide the exact reduction rules, the first one shows that when there are too many yes-instances in the same class, one of them can be safely removed.

\begin{claim}\label{claim:add_pos_assignment}
    Let $x_1$, \ldots, $x_t$ be the assignments of $[n]$, and $a_1$, \ldots, $a_t$ be the respective labels. Let $t \ge k + 2$, $S \subset [n] \setminus [k + 2]$, and for each $i \in [k + 2]$, $a_i = 1$, $x_i$ assigns $i$ to ``true'' and all other variables in $[n] \setminus S$ to ``false''. Let no assignment from $x_1$, \ldots, $x_t$ except for $x_1$ assign $1$ to ``true''. If there exists a $k$-term DNF formula $\varphi$ agreeing with $(x_i, a_i)_{i = 2}^t$, then there exists also a $k$-term DNF formula $\psi$ agreeing with $(x_i, a_i)_{i = 1}^t$.
\end{claim}
\begin{claimproof}
    W.l.o.g. let $\varphi = T_1 \lor \cdots \lor T_k$. If each of $T_1$, \ldots, $T_k$ contains a positive literal, by the pigeonhole principle there exists an assignment $x_i$ out of $x_2$, \ldots, $x_{k + 2}$ that does not satisfy $\varphi$. Therefore, there exists a conjunction $T_j$ that only contains negative literals. If $T_j$ does not contain the negation of variable $1$, $\varphi(x_1) = 1 = \lambda_1$, so setting $\psi = \varphi$ concludes the proof. Otherwise, take $\psi$ to be $\varphi$ where the negation of variable $1$ is removed from the term $T_j$. It is easy to see that that $\psi(a_1) = 1 = \lambda_1$. Moreover, none of $a_2$, \ldots, $a_t$ changes the value of the assignment, since yes-instances cannot be unsatisfied by making the term $T_j$ smaller, and no no-instance become satisfied as only $x_1$ assigns $1$ to ``true''.
\end{claimproof}

The next claim allows us to simplify the negative samples.

\begin{claim}\label{claim:add_neg_assignment}
    Let $x_1$, \ldots, $x_t$ be the assignments of $[n]$, and $a_1$, \ldots, $a_t$ be the respective labels. Let $t \ge 2$, $S \subset [n] \setminus \{1\}$, and let $a_i = 0$, $x_1$ assigns the variable $1$ to ``true'' and all other variables in $[n] \setminus S$ to ``false''. Let no assignment from $x_1$, \ldots, $x_t$ except for $x_1$ assign $1$ to ``true''. If there exists a $k$-term DNF formula $\varphi$ agreeing with $(x_i, a_i)_{i = 2}^t$, then there exists also a $k$-term DNF formula $\psi$ agreeing with $(x_i, a_i)_{i = 1}^t$.
\end{claim}
\begin{claimproof}
    Let $\psi$ be the formula obtained from $\varphi$ by appending the negation of the variable $1$ to each term. For each $i \in [t] \setminus \{1\}$, the assignment $x_i$ keeps its value since it assigns $1$ to ``false''. The assignment $a_{t + 1}$, on the other hand, clearly does not satisfy the constructed formula.
\end{claimproof}

Finally, while the previous two claims reduce the number of samples to consider, the next one reduces the number of variables by removing those that are assigned identically.

\begin{claim}\label{claim:remove_vars}
   Let $x_1$, \ldots, $x_t$ be the assignments of $[n]$, and $a_1$, \ldots, $a_t$ be the respective labels.
   Let variables $1$ and $2$ be assigned identically in the assignments $x_1$, \ldots, $x_t$.
   There exists a $k$-term DNF formula $\varphi$ in variables $2$, \ldots, $n$ agreeing with $(x_i, a_i)_{i = 1}^t$, if and only if there exists a $k$-term DNF formula $\psi$ in variables $1$, \ldots, $n$ agreeing with $(x_i, a_i)_{i = 1}^t$.
\end{claim}
\begin{claimproof}
    In one direction, setting $\psi = \varphi$ clearly satisfies the claim. In the other direction, let $\varphi$ be the formula $\psi$ where in every term a literal of the variable $1$ is replaced with the same literal of the variable $2$. Since no assignment in $x_1$, \ldots, $x_t$ differentiates between $1$ and $2$, $\psi(x_i) = \varphi(x_i)$ for each $i \in [t]$.
\end{claimproof}

We now use Claims~\ref{claim:add_pos_assignment}, \ref{claim:add_neg_assignment}, and \ref{claim:remove_vars} to construct  an instance equivalent to the original with at most $2^s \cdot (k + 2)$ assignment-label pairs on at most $s + 2^s \cdot (k + 2) + 1$ variables. The algorithm simply applies Claims~\ref{claim:add_pos_assignment}, \ref{claim:add_neg_assignment}, and \ref{claim:remove_vars} until none is applicable, possible renumerating the assignments and the variables. We first claim that afterwards no class contains more than $(k + 2)$ assignments. If not, either there is a negatively labelled assignment in the class that assigns some variable outside of $S$ to ``true'', and Claim~\ref{claim:add_neg_assignment} is applicable, or there are at least $(k + 2)$ positive assignments in the class, and Claim~\ref{claim:add_pos_assignment} is applicable, which is a contradiction. Thus, in the constructed instance there are at most $2^s \cdot (k + 2)$ assignments, as desired.

Regarding the variables, observe that outside of $S$, each assignment assigns at most one variable to ``true''. Thus at most $s + 2^s \cdot (k + 2)$ variables assign at least one variable to ``true'', meaning that all other variables are only assigned to ``false''. Since Claim~\ref{claim:remove_vars} is not applicable, there can be at most one such variable.

Now, the algorithm tries all possible terms on the remaining variables, and checks whether each term does not satisfy any negatively-labelled assignment. The algorithm then tries all $k$-term DNF formulas consisting of such terms, and sees if any of them satisfies all positively-labelled assignments.

Regarding the running time, the applications of Claims~\ref{claim:add_pos_assignment}, \ref{claim:add_neg_assignment}, and \ref{claim:remove_vars} are clearly poly-time. In the reduced instance, there are at most $s + 2^s \cdot (k + 2) + 1$ variables, so at most $3^{s + 2^s \cdot (k + 2) + 1}$ possible terms and at most $3^{k \cdot(s + 2^s \cdot (k + 2) + 1)}$ different $k$-term DNF formulas on them. This bounds the size of the hypothesis space by $\fpt(n;k+s)$.
\end{proof}

By applying the aforementioned standard transformation from DNF to CNF, we can obtain an analogoue statement but with every occurrence of ``true'' inverted to ``false''. Another standard transformation which simply negates each assignment in every sample and each literal in the hidden hypothesis then allows us to obtain the same result for all four possible combinations of problems and notions of triviality.
\fi

\begin{thm}
    \label{thm:few_columns}
    \textsc{Learning $k$-Clause CNF} and \textsc{Learning $k$-Term DNF} are both $\fptpactime$ parameterized by $k + s$, where $s$ is the minimum size of a set $S$ of variables such that the support of the distribution satisfies the following:
    \begin{itemize}
    \item each assignment assigns at most one variable outside of $S$ to ``false'' (respectively ``true''), and 
    \item each of the remaining variables, at most one assignment assigns it to ``false'' (respectively ``true'').
    \end{itemize}
    \end{thm}
    
    \ifshort
    \begin{proof}[Proof Sketch]
    By Theorem~\ref{thm:red-learn-to-dec}, it is enough to provide a fixed-parameter algorithm for the consistency checking problem with $f = \fpt(n;k+s)$; in particular, it suffices to argue that we can bound the length of all possible hypotheses by such an $f$. The main technique used in the proof is that of \emph{kernelization}~\cite{CyganFKLMPPS15}, i.e., preprocessing of the instance of the consistency checking problem.
    
    We begin by partitioning the set of labelled samples into equivalence classes, where two samples are equivalent if and only if they agree on $S$. By a series of claims, one can show that if an equivalence class contains more than $k + 2$ samples, one can be safely deleted without changing the instance.
    \end{proof}    
    \fi

\section{Learning on Graphs}
Let $H_1$, \ldots, $H_p$ be a fixed family of graphs. Let $\mathcal{H}$ be a class of graphs that do not contain any of $H_1$, \ldots, $H_p$ as an induced subgraph. In a classical \textsc{$\mathcal{H}$-Vertex Deletion} problem the task is, given a graph $G$ and a parameter $k$, to determine whether there exist a subset of vertices $S \subset V(G)$ of size at most $k$ such that $G - S$ belongs to $\mathcal{H}$. Produced in the standard fashion, the consistency version of this problem receives as input a sequence of graphs $G_1$, \ldots, $G_t$ over the same vertex set $V$, together with the sequence of labels $\lambda_1$, \ldots, $\lambda_t$, and the task is to find a subset $S \subset V$ of size at most $k$ such that $G_i - S \in \mathcal{H}$ if and only if $\lambda_i = 1$. Let us call the respective learning problem \textsc{Learning $\mathcal{H}$-Deletion Set}.

In particular, when the forbidden family consists of a single graph $K_2$, $\mathcal{H}$ is a class of empty graphs, and \textsc{$\mathcal{H}$-Vertex Deletion} is equivalent to \textsc{Vertex Cover}. If the family is $P_3$, the problem becomes \textsc{Cluster Vertex Deletion} --- find a subset of vertices to delete so that the graph turns into a cluster graph, i.e., a disjoint union of cliques.  \textsc{Learning $\mathcal{H}$-Deletion Set} thus generalizes both \textsc{Learning Vertex Cover} and \textsc{Learning Cluster Deletion Set}, where the task is to learn a hidden vertex cover and deletion set to a cluster graph, respectively. 

Let us tie this explicitly to the formal framework of parameterized learning developed above.
In this context, we interpret an element $x \in \{0,1\}^{n}$ as the adjacency matrix of a graph on $N$ vertices, where $n = N^2$.
A concept $c: \{0,1\}^n \rightarrow \{0,1\}$ is represented by a subset $S$ of the $N$ vertices.
The value $c(x)$ indicates whether or not $S$ is an $\mathcal{H}$-Deletion set for the graph with adjacency matrix $x$. 
As noted, the representation scheme of $c$ takes the subset $S$ to $c$ as just described.
Hence, $\mathcal{C}_n$ corresponds to all $\mathcal{H}$-Deletion sets on $N$-vertex graphs.
We parameterize by taking $\mathcal{R}_k$ to be the set of all vertex subsets of size at most $k$,
and let the distributions carry the trivial constant parameterization $\lambda(\mathcal{D}) = 1$.

Note that, for graph problems, the length of hypotheses is always naturally bounded linearly in $n$,
so we will omit explicit references to $f$ for consistency checking.

Next, we show that the well-known $\fpt$-time algorithm~\cite{CyganFKLMPPS15} for \textsc{$\mathcal{H}$-Vertex Deletion} extends to the respective consistency problem, implying that
the wide class of problems characterized by finite family of forbidden induced subgraphs is $\fpt$-time learnable.

\begin{lem} \label{lem:del-consistency}
Let $\mathcal{H}$ be a class of graphs forbidding induced $H_1$, \ldots, $H_p$.    
Then, the parameterized consistency checking problem associated with $\textsc{Learning $\mathcal{H}$-Deletion Set}$ is fixed-parameter tractable.
\end{lem}
\iflong \begin{proof}
Let $(x_1,a_1),\ldots,(x_t,a_t)$ and $k$ be given, and let $G_1,\ldots,G_t$ be the corresponding graphs with $N$ vertices each.
Denote with $\mathcal Y$ be the set of $G_i$ such that $a_i = 1$, and with $\mathcal N$ those where $a_i = 0$.
We need to find a subset $S$ of at most $k$ vertices such that $G - S \in \mathcal{H}$ for all $G \in \mathcal Y$,
and $G - S \notin \mathcal{H}$ for all $G \in\mathcal N$.
First, note that since $\mathcal{H}$ is defined through forbidden induced subgraphs, replacing $S$ with subset of $S$ that still satisfies the first condition will always also still satisfy the second condition. 
It therefore suffices to look only for inclusion-wise minimal subsets that satisfy $G - S \in \mathcal{H}$ for all $G \in \mathcal{Y}$, and then check if $G - S \notin \mathcal{H}$ does indeed hold true for all $G \in \mathcal{N}$.

In order to produce such a minimal set $S$ of size at most $k$, we will in fact enumerate all of them in fixed-parameter time,
by extending the usual bounded search-tree based approach for enumerating minimal vertex covers.
Starting with an empty solution $S = \emptyset$, we iterate over the graphs in $\mathcal{Y}$.

By brute force enumeration, we check if there is an induced copy of any of the $H_i$ in any of them.
Recall that there are a constant number $p$ of graphs $H_i$ of constant size each, say, $q = \max_i |V(H_i)|$,
so that this can be done in time $O(p N^q)$, which is polynomial in $N$. 

If this procedure uncovers an induced copy of $H_i$, then any $\mathcal{H}$-deletion set must remove one of the at most $q$ vertices that the induced copy of $H_i$ occupies in $G$ (if $|S| = k$ at this point, then we discard $S$).
Our enumeration algorithm branches on which of these $q$ choices is made, and adds the respective vertex to $S$ 

If, on the other hand, there was no induced copy of any $H_i$, then $S$ is a minimal $\mathcal{H}$-deletion set for all of the $G \in \mathcal{Y}$, so we check if $G - S \notin \mathcal{H}$ holds true; if so, we output $S$.

As argued, this approach is correct. As for its running time, note that we build up a search-tree of bounded depth $k$ with branching factor at most $q$, and spend at most $O(pN^q)$ time at each step of the recursion,
leading to an overall running time of at most $O(q^k p N^q)$, which is in fixed-parameter tractable in $k$.
\end{proof}\fi

\begin{thm}
    \label{thm:induced_subgraphs}
    Let $\mathcal{H}$ be a class of graphs forbidding induced $H_1$, \ldots, $H_p$. \textsc{Learning $\mathcal{H}$-Deletion Set} is $\fpt$-time PAC-learnable parameterized by the size of the deletion set.
\end{thm}
\iflong \begin{proof}
This follows from Lemma \ref{lem:del-consistency} together with Theorem \ref{thm:red-learn-to-dec}.
\end{proof}\fi

In contrast, we now show that another well-studied family of deletion problems does most likely not admit a positive result similar to the above. 
Namely, for a graph $H$, let $\mathcal{H}$ be a class of graphs that are $H$-minor-free. 
In particular, when $H = K_3$, \textsc{$\mathcal{H}$-Vertex Deletion} is equivalent to \textsc{Feedback Vertex Set}. 
This problem is well-known to admit $\fpt$-time algorithms when parameterized by the size of the vertex set to delete~\cite{CyganFKLMPPS15},
however we show that the consistency checking problem associated with $\textsc{Learning Feedback Vertex Set}$ is $\W[2]$-hard.
That is, unless the two parameterized complexity classes $\FPT$ and $\W[2]$ coincide,
the consistency problem is not in $\FPT$. 
The consequence $\FPT = \W[2]$ is a parameterized analogue to $\P = \NP$, and considered highly unlikely.

\begin{lem} \label{lem:fvs-consis-hardness}
    The consistency checking problem for \textsc{Learning Feedback Vertex Set} is $\W[2]$-hard, 
    even with only yes-instances in the input.
    \label{thm:fvs_hardness}
\end{lem}
\iflong \begin{proof}
We reduce from \textsc{$k$-Hitting Set}. Again, the input is a pair $(U,\mathcal{F})$, where $\mathcal{F}$ is a family $\mathcal{F}$ of subsets $\{F_i\}_{i=1}^m$ such that $F_i \subseteq U$ for all $i$, and $U = \{1,\ldots,n\}$. The task is to produce a set $H \subseteq \{1,\ldots,n\}$ of size at most $k$ such that $F_i \cap H \neq \emptyset$ for all $i$, where $k$ is the parameter.
This problem is known to be $\W[2]$-hard \cite{DowneyF99}.

We construct an instance of the consistency checking problem for $\textsc{Feedback Vertex Set}$ as follows:
The instances $(x_i,a_i)$ with $i=1,\ldots,m$ have $a_i = 1$ for all $i$, and every $x_i$ is a graph $G_i$ on vertices $\{1,\ldots,n\}$, which edges that form a cycle on the vertex set $F_i$.
Every subset of vertices on $\{1,\ldots,n\}$ of size $k$ that is a simultaneous feedback vertex set for all $G_i$ must hit each $F_i$ at least once, and is hence a hitting set of size at most $k$.
Conversely, every simultaneous feedback vertex set for all of the $G_i$ must destroy every cycle,
and hence contain at least one element of $F_i$ for all $i$.

Therefore, there is a feedback vertex set of size at most $k$ consistent with all samples if and only if there is a hitting set of size at most $k$ in the original instance, and hence, consistency checking $\textsc{Feedback Vertex Set}$ is $\W[2]$-hard.
\end{proof}\fi

\begin{thm}
Unless $\FPT = \W[2]$, the learning problem $\textsc{Learning Feedback Vertex Set}$ is not in $\fptpac$.
\label{thm:graphnotPAC}
\end{thm}
\iflong \begin{proof}
Follows from Lemma \ref{lem:fvs-consis-hardness} and Theorem \ref{thm:red-dec-to-learn}.
\end{proof}\fi

On the other hand, the consistency checking problem for \textsc{Learning Feedback Vertex Set} is trivially in $\XP$. Since $\log |\mathcal{H}_{n,k}|$ is trivially upper-bounded by $\bigoh(n)$, by Theorem~\ref{thm:red-learn-to-dec} we immediately obtain:

\begin{obs}
\label{obs:xpgraphs}
\textsc{Learning Feedback Vertex Set} is in $\xppactime$.
\end{obs}

The same argument can be made for learning many other graph structures, such as dominating and independent sets. 

\section{Conclusion}
Over the last two decades, the parameterized complexity paradigm has arguably revolutionized our understanding of computational problems throughout computer science. We firmly believe that applying a more fine-grained, parameterized lens on learning problems can similarly reveal a wealth of research questions and potential breakthroughs that have remained hidden up to now. This article provides researchers with the tools and concepts they need to begin a rigorous exploration of this novel research direction, while also laying the foundations of a bridge that will connect the research communities surrounding learning theory and time complexity. It is perhaps worth noting that these communities have so far remained fairly isolated from each other, with their few interactions occurring at venues dedicated to AI research.

Unlike in typical complexity-theoretic papers, we view the greatest contribution of this article to be the theory-building part, i.e., Section~\ref{sec:fptpac}. Indeed, while the specific algorithms and lower bounds that we use to showcase the theory are non-trivial and interesting in their own right, ensuring that the theoretical foundations fit together, are broad enough to capture potential future parameterizations, but also can support the crucial link to the parameterized consistency checking problem was a highly challenging task.

The introduced parameterized PAC learning framework opens many avenues for future study. For instance, there are interesting and non-trivial examples of parameterized learning problems in the class $\mathsf{PAC}[\mathrm{xp},\mathrm{fpt}]$? 

\section* {Acknowledgements}
The authors acknowledge support from the Austrian Science Foundation (FWF, project Y 1329 START-Programm).
\bibliography{ref}

\end{document}